\documentclass[conference]{IEEEtran}
\IEEEoverridecommandlockouts
\usepackage[usenames,dvipsnames]{xcolor}
\usepackage{amsfonts}
\usepackage{cite} 
\usepackage{amsmath,amsthm} 
\usepackage{amssymb}  
\usepackage{graphicx}
\usepackage{epsfig}
\usepackage{epstopdf}
\usepackage[normal]{subfigure}
\usepackage{mathtools}
\usepackage{etoolbox}
\usepackage{lipsum}
\usepackage{soul}
\usepackage{verbatim}

\usepackage{comment}
\usepackage{algorithm,algorithmic}

\newtheorem{theorem}{\bf Theorem}%[section]
\newtheorem{lemma}{\bf Lemma}
\newtheorem{definition}{\bf Definition}
\newtheorem{corollary}{\bf Corollary}

\newtheorem{remark}{\bf Remark}

\begin{document}

\title{{Generalized Talagrand Inequality for Sinkhorn Distance using Entropy Power Inequality}
}

\author{\IEEEauthorblockN{Shuchan Wang, Photios A. Stavrou and Mikael Skoglund}
\IEEEauthorblockA{\textit{Division of Information Science and Engineering} \\
\textit{KTH Royal Institute of Technology}\\
%Stockholm, Sweden \\
\{shuchan,fstavrou,skoglund\}@kth.se}}
%\and
%\IEEEauthorblockN{Mikael Skoglund}
%\IEEEauthorblockA{\textit{Division of Information Science and Engineering} \\
%\textit{KTH Royal Institute of Technology}\\
%Stockholm, Sweden \\
%fstavrou@kth.se}
%\and
%\and
%\and
%\and

\maketitle

\begin{abstract}
In this paper, we study the connection between entropic optimal transport and entropy power inequality (EPI). First, we prove an HWI-type inequality making use of the infinitesimal displacement convexity of optimal transport map. Second, we derive two Talagrand-type inequalities using the saturation of EPI that corresponds to a numerical term in our expression. We evaluate for a wide variety of distributions this term whereas for Gaussian and i.i.d. Cauchy distributions this term is found in explicit form. We show that our results extend previous results of Gaussian Talagrand inequality for Sinkhorn distance to the strongly log-concave case.  
\end{abstract}

\section{Introduction}\label{sec:intro}
Optimal transport (OT) theory studies how to transport one measure to another in the path with minimal cost. Wasserstein distance is the cost given by the optimal path and closely connected with information measures \cite{talagrand1996transportation, bakry2012dimension, cordero2017transport, bolley2018dimensional, raginsky2018concentration}.

During the last decade, OT has been studied and applied extensively, especially in machine learning community, see, e.g., \cite{zhang2018policy,montavon2016wasserstein,arjovsky2017wasserstein,rigollet2019uncoupled}.  Entropic OT, a technique to approximate the solution of original OT, was given for computational efficiency in \cite{cuturi2013sinkhorn}. A key concept in entropic OT is Sinkhorn distance, which is a generalization of Wasserstein distance with entropic constraint.

One of the applications of OT is in the field of functional inequalities with geometrical content, which includes, for example, Talagrand inequality, HWI inequality, Brunn-Minkowski inequality, etc. Talagrand inequality in \cite{talagrand1996transportation} upper bounds Wasserstein distance by KL-divergence. Recent results in \cite{bakry2012dimension,bolley2018dimensional} obtain several refined Talagrand inequalities with dimensional improvement on multidimensional Eucledian space. These inequalities bound Wasserstein distance with entropy power, which is sharper compared to KL-divergence. Later, \cite{9174478} shows the dimensional improvement can be used in entropic OT and gives a Gaussian Talagrand inequality for Sinkhorn distance. The Talagrand inequalities above can directly give results on measure concentration \cite{raginsky2018concentration}. A strong data processing inequality is also obtained in \cite{9174478} and yields a bound for the capacity of the relay channel.

In this paper, we propose a new approach to study entropic OT using EPIs (for details on EPI see, e.g., \cite{shannon1948mathematical,stam1959some,rioul2010information}), which is to capture the uncertainty caused by entropic constraint by EPI. As a first contribution, we give an HWI-type inequality for Sinkhorn distance by modifying Bolley's proof in \cite{bolley2018dimensional} (see Theorem \ref{th:HWI}). As a second contribution, we derive two new Talagrand-type inequalities (see Theorems \ref{th:1}, \ref{th:ETalagrand2}). These inequalities are obtained via a numerical term related to the saturation (or tightness) of EPI. The value of this term is explicit for Gaussian and i.i.d. Cauchy distributions. We also provide via numerical simulations the computation of the obtained numerical term for a variety of distributions (see Remark \ref{rm:3}). Moreover, we show that Theorem \ref{th:1} coincides with a result by Bolley \cite[Theorem 2.1]{bolley2018dimensional}  (see Corollary \ref{cor:1} and the discussion in Remark \ref{rm:4}) and that Theorem \ref{th:ETalagrand2} recovers the result by Bai et al. \cite[Theorem 2.2]{9174478} (see Remark \ref{cor:2}). Finally, we give a numerical simulation of Theorem \ref{th:1}.

 \section{Notations}
 $\nabla$ is the gradient operator, $\nabla\cdot$ is the divergence operator, $\Delta$ is the Laplacian operator,  $D^2$ is the Hessian operator, $I_n$ is the identity matrix, $Id$ is the identity map, $\|\cdot\|$ is the Euclidean norm, $C^k$ is the set of functions that k-times continuously differentiable. Let $\mathcal{X}, \mathcal{Y}$ be two Polish spaces. Let $\mu$ be a Borel measure on $\mathcal{X}$. For a measurable map $T: \mathcal{X} \rightarrow \mathcal{Y}$, $T_\#\mu$ denotes pushing forward of $\mu$ to $\mathcal{Y}$, i.e. for all $A \subset \mathcal{Y}$, $T_\#\mu[A] = \mu[T^{-1}(A)]$. For $p\geq 1$, $L^p(\mathcal{X})$ denotes the Lebesgue space of $p$-th order for the reference measure $\mu$. We write a random vector $X$ on a Polish space $\mathcal{X}$ in capital letter, an element $x\in\mathcal{X}$ in lower-case letter. $h(\cdot)$, $I(\cdot;\cdot)$, $D(\cdot\|\cdot)$, $J(X)$, $I(\cdot|\cdot)$ denote differential entropy, mutual information, KL-divergence, Fisher information and relative Fisher information, respectively. All the logarithms are natural logarithms. $\exists !$ is unique existence. $*$ is the convolution operator. $\delta(\cdot)$ is the Dirac delta function.  `R.H.S' is the abbreviation of `Right Hand Side'.
 \section{Known Results on OT}
In this section, we state some known results.
\par We start with the definition of Kantorovich problem in OT theory \cite{kantorovich1942translocation}.
\begin{definition}[Kantorovich Problem (KP)] Let $X$ and $Y$ be two random vectors on two Polish spaces $\mathcal{X}$, $\mathcal{Y}$. We denote $\mathcal{P(X)}$ and $\mathcal{P(Y)}$ as the sets of all probability measures on $\mathcal{X}$, $\mathcal{Y}$ respectively. Then $X$ and $Y$ have probability measures $P_X \in \mathcal{P(X)}$, $P_Y \in \mathcal{P(Y)}$. We denote $\Pi(P_X,P_Y)$ as the set of all joint probability measures on $\mathcal{X}\times \mathcal{Y}$ with marginal measures $P_X$, $P_Y$. For a given lower semi-continuous cost function $c(x,y): \mathcal{X}\times \mathcal{Y}\rightarrow \mathbb{R}\cup\{+\infty\}$, Kantorovich problem can be written as:
 \begin{equation}
	\label{eq:KP}
 	\inf_{P \in \Pi(P_X, P_Y)} \mathbb{E}_P[c(X,Y)].
 \end{equation}
\end{definition}

Cuturi in \cite{cuturi2013sinkhorn} gave the concept of entropic OT. In that definition, he adds an information constraint to \eqref{eq:KP}, i.e.,
 \begin{equation}
 	\label{eq:PR}
 	\begin{split}
 	\inf_{P \in \Pi(P_X, P_Y; R)} \mathbb{E}_P[c(X,Y)],
	\end{split}
 \end{equation}
where 
 \begin{equation*}
 \Pi(P_X, P_Y; R) := \{P \in \Pi(P_X, P_Y): I_P(X;Y)\leq R\},
  \end{equation*}
with $I_P(X;Y)$ denoting the mutual information \cite{cover1999elements} between X and Y, and $R$ a non-negative real number. It is well known that the constraint set is convex and compact with respect to the topology of weak convergence \cite[Lemma 4.4]{villani2008optimal}, \cite[1.4]{DupuisPaul2011AWCA}. Using the lower semi-continuity of $c(\cdot,\cdot)$ and the compactness of the constraint set, then, from the extreme value theorem, the minimum in \eqref{eq:PR} is attained.
 
Next, we state Talagrand inequality \cite{talagrand1996transportation}. To do it, we first define Wasserstein distance \cite[Definition 3.4.1]{raginsky2018concentration}. Wasserstein distance is a metric between two measures. Let $d$ be a metric between $x$ and $y$, Wasserstein distance of order $p, p \geq 1$, is defined as follows,
 \begin{equation}
	\label{eq:Was}
 	\mathcal{W}_p(P_X,P_Y) := \inf_{P \in \Pi(P_X, P_Y)} \big\{\mathbb{E}_P[d^p(X,Y)]\big\}^{\frac{1}{p}}.
 \end{equation}
 Similarly, one can define Sinkhorn distance of order $p$ as follows,
 \begin{equation}
 \mathcal{W}_p(P_X,P_Y;R) :=  \inf_{P \in \Pi(P_X, P_Y;R)} \big\{\mathbb{E}_P[d^p(X,Y)]\big\}^{\frac{1}{p}}.
\end{equation}
 
 \begin{theorem}[Talagrand Inequality]\cite[9.3]{villani2003topics}
Let $P_X$ be a reference probability measure with density $e^{-V(x)}$. We say $P_X$ satisfies $\boldsymbol{T}(\lambda)$, i.e., Talagrand inequality with parameter $\lambda > 0$, if for any $P_Y\in \mathcal{P}(\mathcal{Y})$,
\begin{equation}
\label{eq:Talagrand}
\mathcal{W}_2(P_X, P_Y) \leq \sqrt{\frac{2}{\lambda}D(P_Y\|P_X)}.
\end{equation}
\end{theorem}
\begin{remark}
The inequality was originally introduced by Talagrand \cite{talagrand1996transportation} when $P_X$ is Gaussian. Blower \cite{blower2003gaussian}  gave a refinement and proved that strongly log-concavity, i.e., $D^2V \geq \lambda I_n$, leads to $\boldsymbol{T}(\lambda)$.
\end{remark}

Recently, new inequalities with dimensional improvements were obtained. These dimensional improvements were first observed in the Gaussian case of logarithmic Sobolev inequality, Brascamp–Lieb (or Poincaré) inequality \cite{bakry2006logarithmic} and Talagrand inequality \cite{bakry2012dimension}. For a standard Gaussian measure $P_X$, the dimensional Talagrand inequality has the form:
\begin{equation}
\label{eq:Tal1}
\mathcal{W}_2^2(P_X,P_Y)\leq \mathbb{E}[\|Y\|^2] + n - 2ne^{\frac{1}{2n}(\mathbb{E}[\|Y\|^2]-n-2D(P_Y\|P_X))}.
\end{equation}
Bolley et al. in \cite{bolley2018dimensional} generalized the results in \cite{bakry2006logarithmic,bakry2012dimension}  from Gaussian to strongly log-concave or log-concave. Next, we state his result. Let $dP_X = e^{-V}$, where $V: \mathbb{R}^n \rightarrow \mathbb{R}$ is $C^2$ continuous, $D^2V \geq \lambda I_n$. Bolley's dimensional Talagrand inequality is given as follows,
\begin{align}
 \frac{\lambda}{2}\mathcal{W}_2^2(P_X,P_Y) &\leq \mathbb{E}[V(Y)] - \mathbb{E}[V(X)] + n\nonumber\\
&- ne^{\frac{1}{n}(\mathbb{E}[V(Y)] - \mathbb{E}[V(X)]- D(P_Y\|P_X))}\label{eq:Tal3}.
\end{align}
Bai et al. in \cite{9174478} gave a generalization of \eqref{eq:Tal1} to Sinkhorn distance. When $P_X$ is standard Gaussian, 
\begin{equation}
\label{eq:Tal2}
\mathcal{W}^2_2(P_X, P_Y; R) \leq \mathbb{E}[\|Y\|^2] +n  - 2n\sqrt{\frac{1}{2\pi e}(1-e^{-\frac{2}{n}R})}e^{\frac{1}{n}h(Y)}.
\end{equation}
When $R \rightarrow +\infty$, this inequality coincides with \eqref{eq:Tal1}, which is tighter than (\ref{eq:Talagrand}). 
 
\section{Entropy Power Inequality and Deconvolution}
EPI  \cite{shannon1948mathematical} states that for all independent continuous random vectors $X$ and $Y$,
  \begin{equation*}
N(X+Y) \geq N(X)+N(Y),
 \end{equation*}
 where $N(X):= \frac{1}{2\pi e}e^{\frac{2}{n}h(X)}$ denotes the entropy power of $X$.  The equality is achieved when $X$ and $Y$ are Gaussian random vectors with proportional covariance matrices.
 
Deconvolution is a problem of estimating the distribution $dP_{X}$ by the observations $Y_1$,...,$Y_n$ corrupted by additive noice $Z_1$,...,$Z_n$, written as
\begin{equation*}
Y_i = X_i + Z_i,
\end{equation*}
where $X_i$'s are independent and i.i.d distributed in $dP_{X}$, $Z_i$'s are independent and i.i.d distributed in $dP_{Z}$. $X_i$'s and $Z_i$'s are mutual independent. Therefore, their distributions satisfy $dP_{Y} = dP_{X} * dP_{Z}$. Herein, we slightly abuse the concept as simply separating a random vector $Y$ into two independent random vectors $X$ and $Z$. Then their entropies can be bounded by EPI immediately.

Deconvolution is generally a  harder problem than convolution. For instance, log-concave family is convolution stable, i.e., convolution of two log-concave distributions is still log-concave, but we cannot guarantee that deconvolution of two log-concave distributions is still log-concave. A trivial case is the deconvolution of a log-concave distribution by itself is a Dirac function. It should be noted that there are numerical methods to compute deconvolution, see, e.g., \cite{stefanski1990deconvolving,fan1991optimal,masry1991multivariate}.

\section{Main Results}
In this section, we derive our main results. First, we give a new HWI-type inequality.
\begin{theorem}[HWI-type Inequality]
\label{th:HWI}
Let $\mathcal{X}=\mathcal{Y}=\mathbb{R}^n$. Let  $\mu$ be a probability measure with density $e^{-V(x)}$, where $V: \mathbb{R}^n \rightarrow \mathbb{R}$ is $C^2$ continuous, $D^2V \geq \lambda I_n$ with $\lambda > 0$. Let $P_X,P_Y$ be two probability measures on $\mathbb{R}^n$, $P_X,P_Y\ll \mu$. For any independent $Y_1, Y_2$ satisfying $Y_1+Y_2=Y$, $\mathbb{E}[Y_2]=0$ and $h(Y)-h(Y_2)\leq R$, we have
\begin{align}
\label{eq:HWI2}
\frac{\lambda}{2}\mathcal{W}_2^2&(P_X,P_Y;R) \leq \mathbb{E}[V(Y)] - \mathbb{E}[V(X)] + n \nonumber\\
&- n\,e^{\frac{1}{n}(h(Y_1) - h(X))} + \mathcal{W}_2(P_X, P_{Y_1})\sqrt{I(P_{X}|\mu)},
\end{align}
where $I(P |Q):=\mathbb{E}_P[\|\nabla(\log\frac{dP}{dQ})\|^2]$ is relative fisher information.
\end{theorem}
\begin{proof}
See Appendix \ref{pf:1}.
\end{proof}

The next result gives a new Talagrand-type inequality.
\begin{theorem}[Talagrand-type Inequality]
\label{th:1}
Let $\mathcal{X}=\mathcal{Y}=\mathbb{R}^n$. Let $dP_X = e^{-V(x)}dx$, where $V: \mathbb{R}^n \rightarrow \mathbb{R}$ is $C^2$ continuous, $D^2V \geq \lambda I_n$ with $\lambda > 0$, $P_Y \ll P_X$. We have
\begin{align}
	\frac{\lambda}{2}\mathcal{W}_2^2(P_X,P_Y;R) \leq &\mathbb{E}[V(Y)] - \mathbb{E}[V(X)] + n\nonumber\\
	&- nC(P_Y, R)e^{\frac{1}{n}(h(Y) - h(X))},
	\label{eq:ETalagrand}
\end{align}
where $C(P_Y, R) \in [0,1]$ is a numerical term for the given $P_Y$ and $R\geq 0$. 
\end{theorem}
\begin{proof}
Let $dP_X = e^{-V}$ in \eqref{eq:HWI2}. In such case, we have $I(P_{X}|\mu)=0$ from the definition of relative fisher information. Take $C(P_Y, R) = e^{\frac{1}{n}(h(Y_1) - h(Y))}$, then \eqref{eq:ETalagrand} is proved from \eqref{eq:HWI2}. 
\end{proof}
Next, we state some technical remarks on Theorem \ref{th:1}.
\begin{remark}[On Theorem \ref{th:1}]
\label{rm:2}
The equality of \eqref{eq:Taylor} holds when $P_X$ is isotropic Gaussian, i.e., $P_X \sim N(\mu, \sigma^2I_n)$ for some $\mu\in \mathbb{R}^n$ and $\sigma > 0$. The equality in Lemma \ref{lemma:1} holds when $\nabla \varphi$ is affine and $D^2\varphi$ has identical eigenvalues, i.e., $\nabla \varphi = k\cdot Id,\, k \in \mathbb{R}$, see \cite[Lemma 2.6]{bolley2018dimensional}. From \cite[Theorem 1]{janati2020entropic} we know that the linear combination $Y = Y_1 +Y_2$ in Theorem \ref{th:HWI} is the optimizer for entropic OT when $X$ and $Y$ are isotopic Gaussian. In such case, the equality of \eqref{eq:ETalagrand} holds and $C(\cdot,R) = \sqrt{1-e^{-\frac{2}{n}R}}$.
\end{remark}

\begin{remark}[On the numerical term $C(\cdot,\cdot)$]\label{rm:3}
We  observe that $e^{\frac{1}{n}h(Y_1)}$ has the form of a square root of entropy power. Using EPI and the fact that $N(\cdot) \geq 0$, we have
\begin{equation*}
	N(Y) \geq N(Y_1)+N(Y_2) \geq N(Y_1).
\end{equation*}
Therefore $C=e^{\frac{1}{n}(h(Y_1) - h(Y))} = \sqrt{N(Y_1)/N(Y)}\in [0,1]$. When $R = 0$, $Y = Y_2$, the density of $Y_1$ is $\delta(x)$. It means that $e^{\frac{1}{n}h(Y_1)} = 0$, hence $C=0$. When $R = 1$, $Y = Y_1$, $e^{\frac{1}{n}h(Y_1)} = e^{\frac{1}{n}h(Y)}$, hence $C=1$. Now assume that we have $R' < R$. Then $\Pi(P_X, P_Y; R') \subset \Pi(P_X, P_Y; R)$ and the minimization problem \eqref{eq:PR} directly leads to $\mathcal{W}_2(P_X, P_Y; R') \geq \mathcal{W}_2(P_X, P_Y; R)$. Note that $\mathcal{W}_2(P_X, P_Y; R)$ can be also bounded by $C(P_Y, R')$. Hence $C(P_Y, R')\leq C(P_Y, R)$, i.e., $C(\cdot,\cdot)$ is a monotonic non-decreasing function with respect to $R$. Moreover, $C(\cdot, 0) = 0$, $C(\cdot, +\infty) = 1$ for all $P_Y$.

We note that for particular distributions, we may have explicit expression of $C(\cdot,\cdot)$. When $P_Y$ is Gaussian, we can always take the linear combination $Y = Y_1 + Y_2$, where $Y_1$ and $Y_2$ are independent Gaussian and have proportional covariance matrices. EPI is saturated in this case, i.e.,
\begin{equation*}
   \begin{split}
	e^{\frac{2}{n}h(Y_1)} &= e^{\frac{2}{n}h(Y)} - e^{\frac{2}{n}h(Y_2)}\\
	&= (1-e^{-\frac{2}{n}R})e^{\frac{2}{n}h(Y)}.
   \end{split}
   \end{equation*}
Then we have $C(P_Y, R) = e^{\frac{1}{n}(h(Y_1) - h(Y))} = \sqrt{1-e^{-\frac{2}{n}R}}$. For Cauchy distribution $Cauchy(x_0, \gamma)$, its differential entropy is $\log(4\pi \gamma)$. The summation of independent Cauchy random variables $\sum^n_iCauchy(x_i, \gamma_i) \sim Cauchy(\sum^n_ix_i, \sum^n_i\gamma_i)$. When $Y$ is i.i.d. Cauchy, i.e., $(Y)_i\sim Cauchy(x_0, \gamma)$, we take $(Y_1)_i\sim Cauchy(x_0, \frac{1}{4\pi}e^{\frac{1}{n}h(Y)}\cdot (1-e^{-\frac{R}{n}}))$ and $(Y_2)_i\sim Cauchy(0, \frac{1}{4\pi}e^{\frac{1}{n}h(Y)}\cdot e^{-\frac{R}{n}})$. We can see this linear combination satisfies our assumption $h(Y)-h(Y_2)\leq R$ and  $C(P_Y, R) =  e^{\frac{1}{n}(h(Y_1) - h(Y))}=1-e^{-\frac{R}{n}}$.

Note that the linear combination $Y = Y_1+Y_2$ is not unique, according to the assumption of Theorem \ref{th:HWI}. Consequently, it leads to the non-uniqueness of $C(\cdot,\cdot)$. In order to obtain the tightest bound in \eqref{eq:ETalagrand}, the optimal $C^*(P_Y, R) = \sup e^{\frac{1}{n}(h(Y_1) - h(Y))}$ subject to $Y_1+Y_2=Y$ and $h(Y)-h(Y_2)\leq R$. To look into this optimization problem, we introduce Courtade's reverse EPI \cite{courtade2017strong} as follows. If we have independent $X$ and $Y$ with finite second moments and choose $\theta$ to satisfy $\theta/(1-\theta)=N(Y)/N(X)$, then
\begin{equation}
\label{eq:repi}
N(X+Y) \leq (N(X) + N(Y))(\theta p(X) + (1-\theta) p(Y)),
\end{equation}
where $p(X) := \frac{1}{n}N(X)J(X)\geq 1$ is Stam defect. $p(X)$ is affine invariant, i.e. $p(X) = p(tX)$, $t>0$ because $t^2N(X) = N(tX)$ and $J(X) = t^2J(X)$. The equality $p(X)=1$ holds only if $X$ is Gaussian. In our case, $\theta=N(Y_1)/(N(Y_1)+ N(Y_2))$. When $\theta \rightarrow 1$, (\ref{eq:repi}) becomes
\begin{equation*}
N(Y) \lesssim (N(Y_1)+ N(Y_2))\cdot p(Y_2).
\end{equation*}
It means the saturation of EPI is controlled by $p(Y_2)$ when the noise $Y_2$ is small, i.e., $R$ is large. In this case, $C^*(P_Y, R)  \approx \sqrt{1-e^{-\frac{2}{n}R}}$ if we let $Y_2$ close to Gaussian, i.e., $p(Y_2) =1$.  On the other hand, when $\theta \rightarrow 0$, EPI can also be saturated if we let $Y_1$ close to Gaussian.

In Fig. \ref{fig:Fig1}, we illustrate the numerical simulations of $C(\cdot,\cdot)$. For general distributions beyond Gaussian and i.i.d. Cauchy, we approximate $C(\cdot,\cdot)$ using kernel methods of deconvolution, see, e.g., \cite{stefanski1990deconvolving,fan1991optimal}. Our strategy of deconvolution in Fig. \ref{fig:Fig1} is to let $Y_2 = tY'$, where $Y'$ is a copy of $Y$ and $t \in [0,1]$. Gaussian mixture is an exception for this strategy because its spectrum would be not integrable. Instead we let $Y_2$ to be Gaussian for Gaussian mixture. We note this strategy is mostly not optimal and the optimal way to maximize the entropy power above is still an open question. 
\end{remark}
\begin{figure}[h]
\begin{center}
\includegraphics[trim=1.0cm 10.3cm 0.5cm 9.2cm, clip=true, totalheight=0.18\textheight,width=0.5\textwidth]{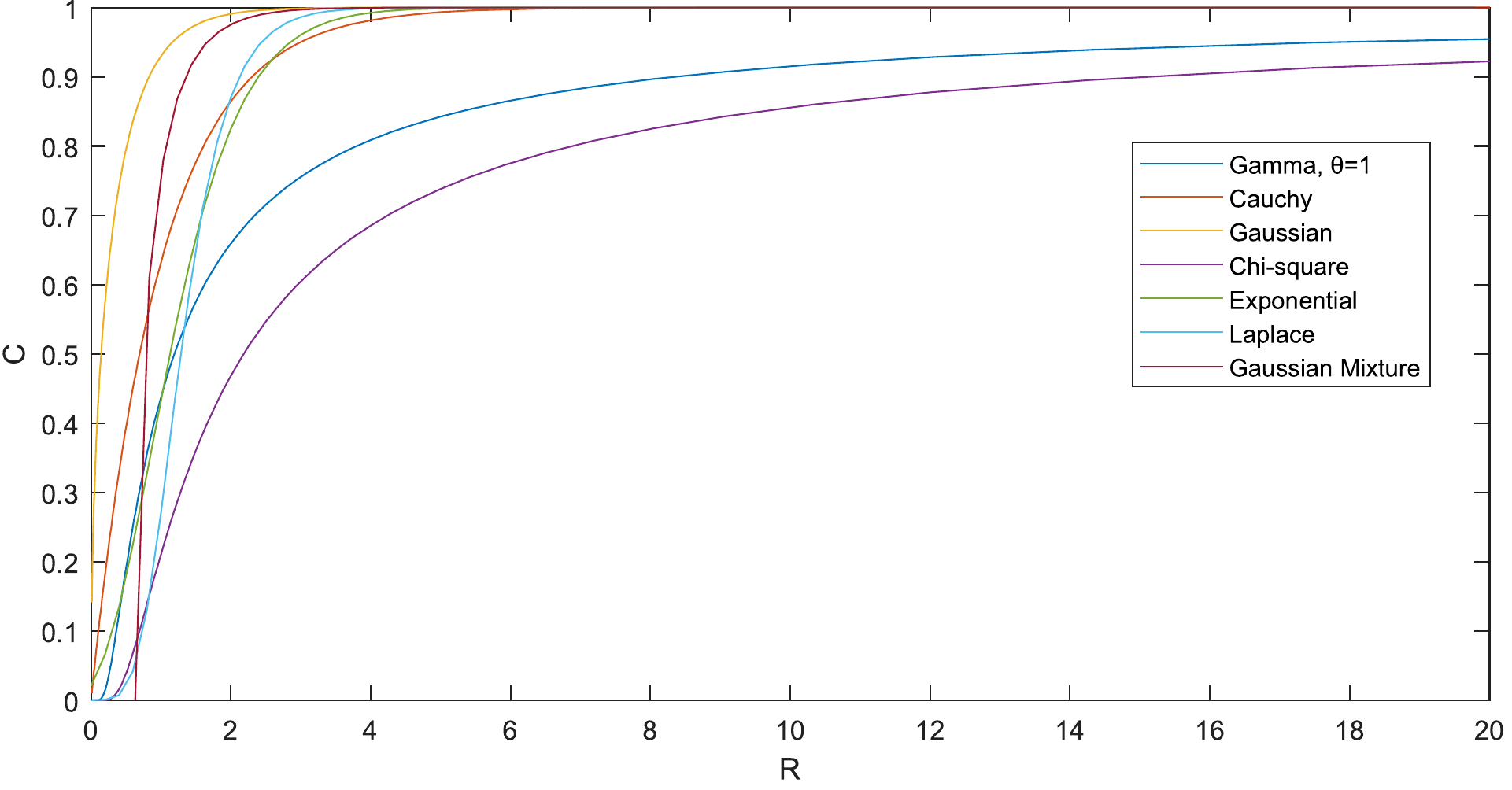}
\end{center}
\vspace{-0.6cm}
\caption {Plot of the numerical term C subject to the information constraint R  evaluated with respect to different distributions for the one dimensional case.}
\label{fig:Fig1}
\end{figure}

The following corollary is immediate from Theorem \ref{th:1}.
\begin{corollary}
\label{cor:1}
Wasserstein distance is bounded by
\begin{equation}
\label{eq:uncon}
	\frac{\lambda}{2}\mathcal{W}_2^2(P_X,P_Y) \leq \mathbb{E}[V(Y)] - \mathbb{E}[V(X)] + n - ne^{\frac{1}{n}(h(Y) - h(X))}.
\end{equation}
\end{corollary}
\begin{proof}
This is immediate from Theorem \ref{th:1} when $R \rightarrow \infty$.
\end{proof}
\begin{remark}
\label{rm:4}
We note that (\ref{eq:uncon}) is equivalent to Bolley's dimensional Talagrand inequality \eqref{eq:Tal3} and it is tighter than the classical Talagrand inequality (\ref{eq:Talagrand}). Notice under our assumptions, $h(X) = \mathbb{E}[V(X)]$ and $D(P_Y\|P_X) = \mathbb{E}[V(Y)] - h(Y)$ because $dP_X = e^{-V(x)}dx$. Clearly, by substituting these expressions to the last term of \eqref{eq:uncon} we obtain \eqref{eq:Tal3}. Since $e^\mu \geq 1 + \mu$, (\ref{eq:Tal3}) is, in general, tighter than the classical Talagrand inequality (\ref{eq:Talagrand}), i.e., R.H.S. of (\ref{eq:Tal3}) $\leq$ R.H.S. of (\ref{eq:Talagrand}). The equality holds if and only if $h(Y) = h(X)$.
\end{remark}
\begin{remark}[On measure concentration]\label{rm:concen}
We notice that $C$ is the only difference between \eqref{eq:Tal3} and \eqref{eq:ETalagrand}, from Remark \ref{rm:4}. Therefore, we can immediately get a result of measure concentration following \cite[Corollary 2.4]{bolley2018dimensional}. 

Here we state the result of measure concentration obtained from \eqref{eq:ETalagrand}. Let $d\mu = e^{-V}$, where $V: \mathbb{R}^n \rightarrow \mathbb{R}$ is $C^2$ continuous, $D^2V \geq \lambda I_n$ with $\lambda > 0$.  Let $A \subset \mathbb{R}^n$, $A_r := \{x\in \mathbb{R}^n| \forall y\in A, \|x-y\|> r\}$ for $r\geq 0$ and $c_A := \sqrt{2\lambda^{-1}\log(1/\mu(A))}$. Then for $r\geq c_A$,
\begin{equation}\label{eq:concen}
\mu(A_r) \leq C^{-n}\cdot e^{-\frac{\lambda}{2}(r - c_A)^2}.
\end{equation}
This inequality obtained from Sinkhorn distance is dimension dependent, compared to the dimension-free one in \cite{bolley2018dimensional}. Note that $C\in [0,1]$. Hence, the increment of dimension leads to a slower concentration, i.e., a looser bound in \eqref{eq:concen}.
\end{remark}

The next theorem is another Talagrand-type inequality. Compared to Theorem 3, the following theorem is a bound obtained using a term related to the saturation of $P_X$ instead of the saturation of $P_Y$ that was used in Theorem \ref{th:1}.
\begin{theorem}
\label{th:ETalagrand2}
Let $\mathcal{X}=\mathcal{Y}=\mathbb{R}^n$. Without loss of generality, let $X$ be a zero-mean random vector with density $e^{-V(x)}$, where $V: \mathbb{R}^n \rightarrow \mathbb{R}$ is $C^2$ continuous, $D^2V \geq \lambda I_n$ with $\lambda > 0$, $P_Y \ll P_X$. Then we have
\begin{align}
	\frac{\lambda}{2}\mathcal{W}_2^2(P_X,P_Y;R) &\leq \mathbb{E}[V(Y)] - \mathbb{E}[V(X)] + n\nonumber\\
	 &- nC_x(P_X, R)e^{\frac{1}{n}(h(Y) - h(X))} + \epsilon,
	\label{eq:ETalagrand2}
\end{align}
where $\epsilon$ is a term related to the linearity of $V$.
\end{theorem}
\begin{proof}
See Appendix \ref{pf:3}.
\end{proof}
We make the following technical comments on Theorem \ref{th:ETalagrand2}.
\begin{remark}[On Theorem \ref{th:ETalagrand2}]
\label{rm:5}
Similar to $C(P_Y, R)$, $C_x(P_X, R)$ is also related to the saturation of EPI, as shown in the proof.  However, \eqref{eq:ETalagrand2}  is less natural than \eqref{eq:ETalagrand} because of the extra term $\epsilon$. When $\nabla V$ is nearly linear, $\epsilon$ should be small. When $\nabla V$ is far from linear, $\epsilon$ is unknown. 
\end{remark}

\begin{remark}\label{cor:2}
When $\nabla V$ is linear, $\epsilon$ is zero and $C_x(\cdot,R) = \sqrt{1-e^{-\frac{2}{n}R}}$, as simply taking $t = \sqrt{1-e^{-\frac{2}{n}R}}$ in the proof. In such case, \eqref{eq:ETalagrand2} recovers \eqref{eq:Tal2} by taking $X$ as a standard Gaussian, i.e. $V(x) = \|x\|^2/2 + k$, where $k$ is a normalization factor. Substitute $V$ and times 2 on both sides of (\ref{eq:ETalagrand2}), we have
  \begin{align}
	&\mathcal{W}_2^2(X,Y; R)\nonumber\\ \leq& \mathbb{E}[\|Y\|^2] - \mathbb{E}[\|X\|^2] + 2n - 2n\sqrt{1-e^{-\frac{2}{n}R}}e^{\frac{1}{n}(h(Y) - h(X))}\nonumber\\
	=& \mathbb{E}[\|Y\|^2] +n  - 2n\sqrt{\frac{1}{2\pi e}(1-e^{-\frac{2}{n}R})}e^{\frac{1}{n}h(Y)}\label{eq:compare}.
 \end{align}

\end{remark}

\section{Simulation and Discussion}
We simulate a relatively tight situation for bound \eqref{eq:ETalagrand} in Fig. \ref{fig:strLog}, using \cite{flamary2021pot}. The bound can be loose in some scenarios, e.g., the term $C$ is not optimal, or $P_Y$ is not absolutely continuous to $P_X$, which is, for instance, a reason of mode collapse in GAN training \cite{arjovsky2017wasserstein}. From Remark \ref{rm:concen}, it can be also seen that the information constraint causes more smoothing on higher dimension for the original Wasserstein distance. We know that, in machine learning, samples are usually embeded in a low dimension manifold with the disturbance of a high dimension noise. The dimensionality provides a wider boundary for decision for small disturbance on high dimension and may be one of the explanation that Sinkhorn distance outperform other metrics in \cite{cuturi2013sinkhorn}.
\begin{figure}[h]
\centering
        \includegraphics[trim=1cm 20cm 3.5cm 1cm, clip=true, totalheight=0.18\textheight,width=0.5\textwidth]{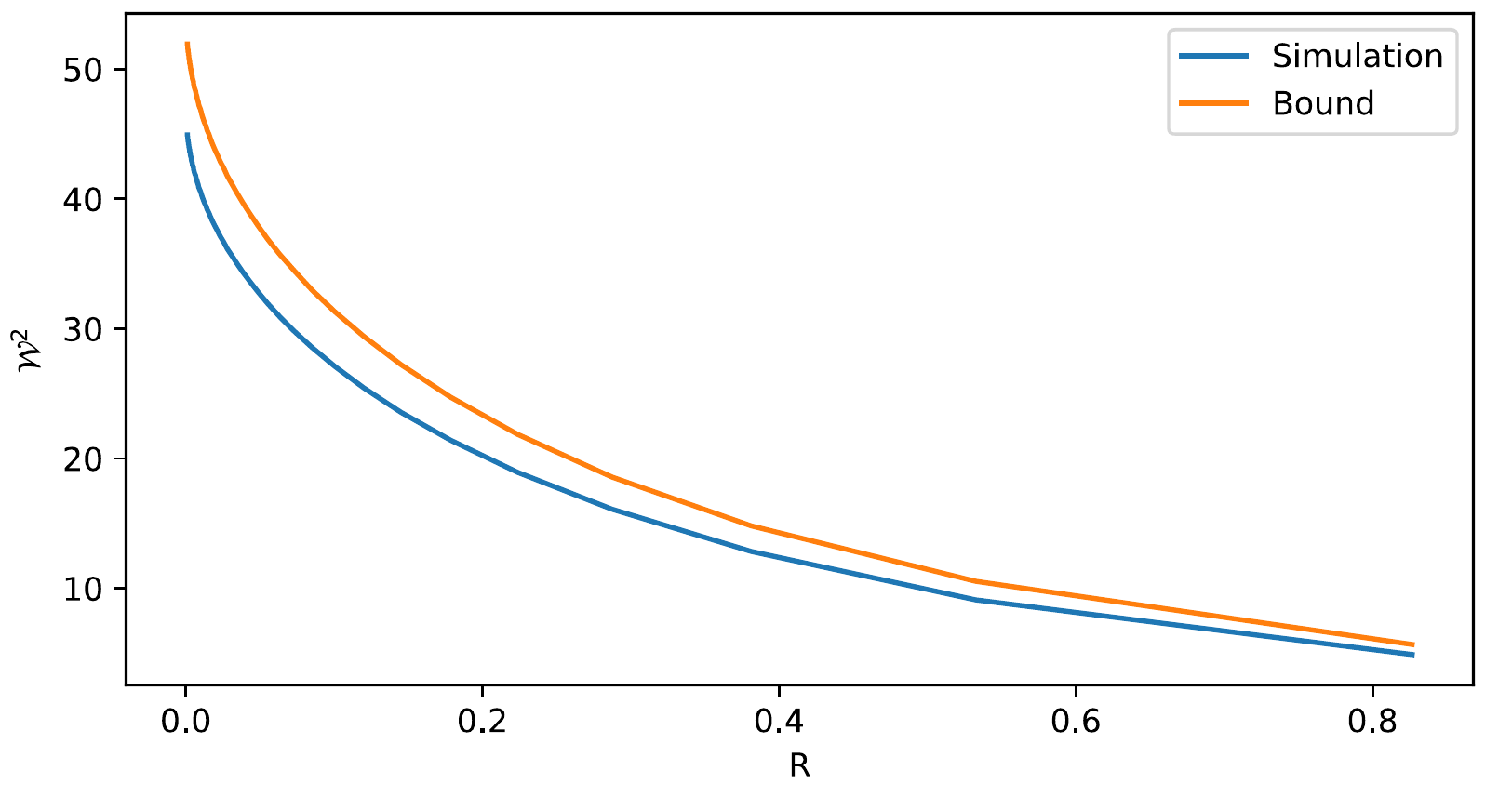}
        \vspace{-0.8cm}
        \caption {Bound \eqref{eq:ETalagrand} for $dP_X = e^{-V}, V = (x/5)^2/2 + |x/10| + e^{-|x/10|} +k, k \in \mathbb{R}$ and $dP_Y \sim \mathcal{N}(0,\frac{1}{25})$.}
        \label{fig:strLog}
\end{figure}

\appendices
\section{Proof of Theorem \ref{th:HWI}}\label{pf:1}
For a $C^2$ continuous function $V: \mathbb{R}^n \rightarrow \mathbb{R}$, $D^2V \geq \lambda I_n$, by Taylor formula \cite[Lemma 2.5]{bolley2018dimensional}, there exists a $t\in [0,1]$ satisfying
\begin{align}
&V(y) - V(x)\nonumber \\
=& \nabla V(x)\cdot (y-x) + (y-x)\cdot D^2V(tx+(1-t)y) (y-x)/2\nonumber\\
\geq& \nabla V(x)\cdot (y-x) + \frac{\lambda}{2}\|y-x\|^2\label{eq:Taylor}.
\end{align}
Hence we can bound the second order cost by
\begin{align}
&\frac{\lambda}{2}\int_{\mathcal{X}\times\mathcal{Y}}\|y-x\|^2\,dP\nonumber\\
 \leq &\int_{\mathcal{X}\times\mathcal{Y}}V(y) - V(x) - \nabla V(x)\cdot (y-x)\,dP\label{eq:AA1}.
\end{align}
Because entropic OT is a minimization problem, we can take any case in $\Pi(P_X,P_Y;R)$ to bound $\mathcal{W}_2(P_X,P_Y;R)$. We take a linear combination $Y = Y_1+Y_2$, where $Y_1$ and $Y_2$ are independent, $h(Y)-h(Y_2)\leq R$ and $\mathbb{E}[Y_2]=0$. Assume there is a Brenier map between $Y_1$ and $X$, i.e., $Y_1 = \nabla\varphi(X)$, which always exists, according to Theorem \ref{Brenier} (see Appendix \ref{apd:1}). Then, we can see this special case is in $\Pi(P_X,P_Y,R)$, namely,
\begin{align}
	I_P(X;Y) &= h(Y) - h(Y|X)\nonumber\\
	&= h(Y)- h(Y_1+Y_2|X)\nonumber\\
	&=  h(Y) - h(Y_2)\nonumber\\
	&\leq R\nonumber.
\end{align}
Let $d\mu = e^{-V}$, where $V: \mathbb{R}^n \rightarrow \mathbb{R}$ is $C^2$ continuous, $D^2V \geq \lambda I_n$. In order to bound Sinkhorn distance, we just need to bound $\int_{\mathcal{X}\times\mathcal{Y}}\nabla V(x)\cdot (y-x)\,dP$, according to \eqref{eq:AA1}. This term can be bounded as follows,
\begin{align}
  	 &\int_{\mathcal{X}\times\mathcal{Y}}\nabla V(x)\cdot (y-x)\,dP\nonumber\\
	 =& \iint \nabla V(x)\cdot (\nabla \varphi(x) + y_2 - x)\,dP_{Y_2}\,dP_{X}\nonumber\\
	 =& \int \nabla V(x)\cdot (\nabla \varphi(x) - x)\,dP_{X}\nonumber\\
	 =& \int \nabla V(x)\cdot (\nabla \varphi(x) - x)\frac{dP_{X}}{d\mu}\,d\mu\nonumber\\
	 \geq& \int \Delta\varphi(x) f \,d\mu - n + \int (\nabla\varphi(x) -x)\cdot \nabla f\,d\mu\label{eq:pf2}\\
	 \geq& ne^{\frac{1}{n}(h(Y_1) - h(X))}-n - \mathcal{W}_2(P_{X}, P_{Y_1})\cdot\sqrt{I(P_{X}|\mu)} \nonumber,
\end{align}
where we take the Radon-Nikodym derivative $f = \frac{dP_{X}}{d\mu}$ in \eqref{eq:pf2} and apply Lemma \ref{lemma:1} in Appendix \ref{apd:1}. This completes the derivation.

 \section{Proof of Theorem \ref{th:ETalagrand2}}\label{pf:3}
Let $X'$ be a copy of $X$. $X$ can be written as a linear combination $X = tX' + X_2$, where $X_2$ is zero-mean and independent with $X'$, $h(X)-h(X_2)\leq R$, $t\in[0,1]$. Assume there exists a Brenier map $Y = \nabla\varphi(X')$. Similar to the proof of Theorem $\ref{th:1}$, this case is also in $\Pi(P_X,P_Y,R)$. Then we have
\begin{align}
 &\int_{\mathcal{X}\times\mathcal{Y}}\nabla V(x)\cdot (y-x)\,dP\nonumber\\
 = &\iint \nabla V(tx'+ x_2)\cdot \nabla \varphi(x')\,dP_{X'}\,dP_{X_2}-n\label{eq:4:3}\\
= &\iint (\nabla V(tx'+ x_2) - t\cdot\nabla V(x'))\cdot \nabla \varphi(x')\,dP_{X'}\,dP_{X_2}\nonumber\\&+ t\int \nabla V(x')\cdot \nabla \varphi(x')\,dP_{X'}-n\nonumber\\
= &\iint (\nabla V(tx'+ x_2) - t\cdot\nabla V(x'))\cdot \nabla \varphi(x')\,dP_{X'}\,dP_{X_2}\nonumber \\& + t\int\Delta \varphi \,dP_{X'}-n\label{eq:4:2}\\
\geq & t\cdot ne^{\frac{1}{n}(h(Y) - h(X))} - \epsilon-n\label{eq:4:1},
 \end{align}
where we use Lemma \ref{lemma:2} of Appendix \ref{apd:1} in \eqref{eq:4:3} and \eqref{eq:4:2}. In \eqref{eq:4:1}, we let $\epsilon = -\iint (\nabla V(tx'+ x_2) - t\cdot\nabla V(x'))\cdot \nabla \varphi(x')\,dP_{X'}\,dP_{X_2}$ and apply \eqref{eq:1:1}. After changing the order of integral, we can see that $\int\nabla V(tx'+ x_2)\,dP_{X_2}$ is a smoothed version of $\nabla V(tx')$. When $\nabla V$ is a linear function perturbed by a zero mean noise, i.e., $\nabla V(tx) = t\cdot\nabla V(x) + W$, the integral of $x_2$ is cancelled out and $\epsilon = 0$. Take $C_x(P_X, R)=t$, then we finish the proof.
 
 \section{Useful Theorem and Lemmas}
 \label{apd:1}
  \begin{theorem}[Brenier]\cite[Theorem 2.12]{villani2003topics}
 \label{Brenier}
Let $P_X \in \mathcal{P}(X)$, $P_Y \in \mathcal{P}(Y)$ with $X \subset \mathbb{R}^n$, $Y \subset \mathbb{R}^n$  and  assume that $dP_X$, $dP_Y$ both have finite second moments. Then, for Kantorovich problem with cost $c(x,y) = \frac{1}{2}\|x-y\|^2$, $\exists ! \varphi: X \rightarrow \mathbb{R}$ gives the optimal coupling
\begin{equation*}
P^* = (Id \times \nabla \varphi)_{\#}P_X,
\end{equation*}
 where $\varphi$ is convex.
 \end{theorem}
\begin{lemma}\cite[Theorem 9.17]{villani2003topics}
\label{lemma:1}
Let $d\mu = e^{-V}$. Let $f =  \frac{dP_{X}}{d\mu}$ being a Radon-Nikodym derivative between two measures $P_X$ and $\mu$. Let $\nabla \varphi$ be a Brenier map as in Theorem \ref{Brenier}. We have
\begin{align}
&\int \nabla V(x)\cdot (\nabla \varphi(x) -x)f(x)d\mu(x)\nonumber\\
\geq& \int [(\Delta \varphi - n)f+(\nabla \varphi- x)\cdot\nabla f]\,d\mu\nonumber\\
=&\int \Delta\varphi f \,d\mu - n + \int (\nabla\varphi-x)\cdot \nabla f\,d\mu,\label{eq:lemma1}
\end{align}
where $\nabla \varphi(x) -x$ is called displacement. For the first term of \eqref{eq:lemma1}, because $\varphi$ is convex, from \cite[Lemma 2.6]{bolley2018dimensional}, we have
 \begin{equation}
 \label{eq:1:1}
\int \Delta\varphi \,dP_X  \geq ne^{\frac{1}{n}(h(\nabla\varphi(X))-h(X))}.
 \end{equation}
Moreover, the last term of \eqref{eq:lemma1} can be bounded using Cauchy–Schwarz inequality as follows,
 \begin{align*}
&\int (\nabla\varphi-x)\cdot \nabla f\,d\mu\\ \geq& -\bigg[\int \|\nabla\varphi-x\|^2f\,d\mu\bigg]^{1/2}\bigg[\int\frac{\|\nabla f\|^2}{f}d\mu\bigg]^{1/2}\\
= &-\mathcal{W}_2(P_{X}, \nabla\varphi_{\#}P_{X})\cdot\sqrt{I(P_{X}|\mu)}.
\end{align*}
 \end{lemma}
\begin{lemma} \cite[Fact 7]{cordero2017transport} \label{lemma:2}
For any $\nabla \varphi \in L^1(\mathcal{X})\cap L^2(\mathcal{X})$ on a Polish space $(\mathcal{X}, \mu)$ and $d\mu = e^{-V}$, we have
\[\int\Delta \varphi \,d\mu = \int \nabla \varphi \cdot \nabla V \,d\mu.\]
\end{lemma}

\section*{Acknowledgement}
This work is funded in part by the Swedish Foundation for Strategic Research.

\IEEEtriggeratref{14}
\bibliographystyle{IEEEtran}
\bibliography{cite}

\begin{thebibliography}{10}
\providecommand{\url}[1]{#1}
\csname url@samestyle\endcsname
\providecommand{\newblock}{\relax}
\providecommand{\bibinfo}[2]{#2}
\providecommand{\BIBentrySTDinterwordspacing}{\spaceskip=0pt\relax}
\providecommand{\BIBentryALTinterwordstretchfactor}{4}
\providecommand{\BIBentryALTinterwordspacing}{\spaceskip=\fontdimen2\font plus
\BIBentryALTinterwordstretchfactor\fontdimen3\font minus
  \fontdimen4\font\relax}
\providecommand{\BIBforeignlanguage}[2]{{%
\expandafter\ifx\csname l@#1\endcsname\relax
\typeout{** WARNING: IEEEtran.bst: No hyphenation pattern has been}%
\typeout{** loaded for the language `#1'. Using the pattern for}%
\typeout{** the default language instead.}%
\else
\language=\csname l@#1\endcsname
\fi
#2}}
\providecommand{\BIBdecl}{\relax}
\BIBdecl

\bibitem{talagrand1996transportation}
M.~Talagrand, ``Transportation cost for gaussian and other product measures,''
  \emph{Geometric \& Functional Analysis GAFA}, vol.~6, no.~3, pp. 587--600,
  1996.

\bibitem{bakry2012dimension}
D.~Bakry, F.~Bolley, and I.~Gentil, ``Dimension dependent hypercontractivity
  for gaussian kernels,'' \emph{Probability Theory and Related Fields}, vol.
  154, no.~3, pp. 845--874, 2012.

\bibitem{cordero2017transport}
D.~Cordero-Erausquin, ``Transport inequalities for log-concave measures,
  quantitative forms, and applications,'' \emph{Canadian Journal of
  Mathematics}, vol.~69, no.~3, pp. 481--501, 2017.

\bibitem{bolley2018dimensional}
F.~Bolley, I.~Gentil, A.~Guillin \emph{et~al.}, ``Dimensional improvements of
  the logarithmic sobolev, talagrand and brascamp--lieb inequalities,''
  \emph{The Annals of Probability}, vol.~46, no.~1, pp. 261--301, 2018.

\bibitem{raginsky2018concentration}
M.~Raginsky and I.~Sason, ``Concentration of measure inequalities in
  information theory, communications and coding,'' \emph{Foundations and Trends
  in Communications and Information Theory; NOW Publishers: Boston, MA, USA},
  2018.

\bibitem{zhang2018policy}
R.~Zhang, C.~Chen, C.~Li, and L.~Carin, ``Policy optimization as wasserstein
  gradient flows,'' in \emph{International Conference on Machine
  Learning}.\hskip 1em plus 0.5em minus 0.4em\relax PMLR, 2018, pp. 5737--5746.

\bibitem{montavon2016wasserstein}
G.~Montavon, K.-R. M{\"u}ller, and M.~Cuturi, ``Wasserstein training of
  restricted boltzmann machines,'' in \emph{Proceedings of the 30th
  International Conference on Neural Information Processing Systems}, 2016, pp.
  3718--3726.

\bibitem{arjovsky2017wasserstein}
M.~Arjovsky, S.~Chintala, and L.~Bottou, ``Wasserstein generative adversarial
  networks,'' in \emph{International conference on machine learning}.\hskip 1em
  plus 0.5em minus 0.4em\relax PMLR, 2017, pp. 214--223.

\bibitem{rigollet2019uncoupled}
P.~Rigollet and J.~Weed, ``Uncoupled isotonic regression via minimum
  wasserstein deconvolution,'' \emph{Information and Inference: A Journal of
  the IMA}, vol.~8, no.~4, pp. 691--717, 2019.

\bibitem{cuturi2013sinkhorn}
M.~Cuturi, ``Sinkhorn distances: lightspeed computation of optimal transport.''
  in \emph{NIPS}, vol.~2, no.~3, 2013, p.~4.

\bibitem{9174478}
Y.~{Bai}, X.~{Wu}, and A.~{{\"O}zg{\"u}r}, ``Information constrained optimal
  transport: From talagrand, to marton, to cover,'' in \emph{2020 IEEE
  International Symposium on Information Theory (ISIT)}, 2020, pp. 2210--2215.

\bibitem{shannon1948mathematical}
C.~E. Shannon, ``A mathematical theory of communication,'' \emph{The Bell
  system technical journal}, vol.~27, no.~3, pp. 379--423, 1948.

\bibitem{stam1959some}
A.~J. Stam, ``Some inequalities satisfied by the quantities of information of
  fisher and shannon,'' \emph{Information and Control}, vol.~2, no.~2, pp.
  101--112, 1959.

\bibitem{rioul2010information}
O.~Rioul, ``Information theoretic proofs of entropy power inequalities,''
  \emph{IEEE Transactions on Information Theory}, vol.~57, no.~1, pp. 33--55,
  2010.

\bibitem{kantorovich1942translocation}
L.~V. Kantorovich, ``On the translocation of masses,'' in \emph{Dokl. Akad.
  Nauk. USSR (NS)}, vol.~37, 1942, pp. 199--201.

\bibitem{cover1999elements}
T.~M. Cover, \emph{Elements of information theory}.\hskip 1em plus 0.5em minus
  0.4em\relax John Wiley \& Sons, 1999.

\bibitem{villani2008optimal}
C.~Villani, \emph{Optimal transport: old and new}.\hskip 1em plus 0.5em minus
  0.4em\relax Springer Science \& Business Media, 2008, vol. 338.

\bibitem{DupuisPaul2011AWCA}
P.~Dupuis and R.~S. Ellis, \emph{\BIBforeignlanguage{eng}{A Weak Convergence
  Approach to the Theory of Large Deviations}}, 1st~ed., ser. Wiley series in
  probability and statistics.\hskip 1em plus 0.5em minus 0.4em\relax Hoboken:
  Wiley-Interscience, 2011.

\bibitem{villani2003topics}
C.~Villani, \emph{Topics in optimal transportation}.\hskip 1em plus 0.5em minus
  0.4em\relax American Mathematical Soc., 2003, no.~58.

\bibitem{blower2003gaussian}
G.~Blower, ``The gaussian isoperimetric inequality and transportation,''
  \emph{Positivity}, vol.~7, no.~3, pp. 203--224, 2003.

\bibitem{bakry2006logarithmic}
D.~Bakry, M.~Ledoux \emph{et~al.}, ``A logarithmic sobolev form of the li-yau
  parabolic inequality,'' \emph{Revista Matem{\'a}tica Iberoamericana},
  vol.~22, no.~2, pp. 683--702, 2006.

\bibitem{stefanski1990deconvolving}
L.~A. Stefanski and R.~J. Carroll, ``Deconvolving kernel density estimators,''
  \emph{Statistics}, vol.~21, no.~2, pp. 169--184, 1990.

\bibitem{fan1991optimal}
J.~Fan, ``On the optimal rates of convergence for nonparametric deconvolution
  problems,'' \emph{The Annals of Statistics}, pp. 1257--1272, 1991.

\bibitem{masry1991multivariate}
E.~Masry, ``Multivariate probability density deconvolution for stationary
  random processes,'' \emph{IEEE Transactions on Information Theory}, vol.~37,
  no.~4, pp. 1105--1115, 1991.

\bibitem{janati2020entropic}
H.~Janati, B.~Muzellec, G.~Peyr{\'e}, and M.~Cuturi, ``Entropic optimal
  transport between unbalanced gaussian measures has a closed form,''
  \emph{Advances in Neural Information Processing Systems}, vol.~33, 2020.

\bibitem{courtade2017strong}
T.~A. Courtade, ``A strong entropy power inequality,'' \emph{IEEE Transactions
  on Information Theory}, vol.~64, no.~4, pp. 2173--2192, 2017.

\bibitem{flamary2021pot}
\BIBentryALTinterwordspacing
R.~Flamary, N.~Courty, A.~Gramfort, M.~Z. Alaya, A.~Boisbunon, S.~Chambon,
  L.~Chapel, A.~Corenflos, K.~Fatras, N.~Fournier, L.~Gautheron, N.~T. Gayraud,
  H.~Janati, A.~Rakotomamonjy, I.~Redko, A.~Rolet, A.~Schutz, V.~Seguy, D.~J.
  Sutherland, R.~Tavenard, A.~Tong, and T.~Vayer, ``Pot: Python optimal
  transport,'' \emph{Journal of Machine Learning Research}, vol.~22, no.~78,
  pp. 1--8, 2021. [Online]. Available:
  \url{http://jmlr.org/papers/v22/20-451.html}
\BIBentrySTDinterwordspacing

\end{thebibliography}
\end{document}